\documentclass{article} 
\usepackage[top=1.2in, bottom=1.2in, left=1.2in, right=1.2in]{geometry}

\usepackage{amsmath,amssymb,amsfonts,amsthm}
\allowdisplaybreaks[1] 

\usepackage{color}
\usepackage[usenames,dvipsnames]{xcolor}

\usepackage{cite}	
\usepackage{url}

\newtheorem{theorem}{Theorem}
\newtheorem{remark}{Remark}

\theoremstyle{definition}
\newtheorem{definition}{Definition}
\theoremstyle{definition}

\usepackage{hyperref}
\hypersetup{
   colorlinks=true,
   citecolor=Green
}

\usepackage{algorithmic}
\usepackage[ruled, noend,lined,linesnumbered]{algorithm2e} 

\newcommand\domain{\mathcal{X}}

\begin{document}
\title{Differentially Private Release of Public Transport Data: The Opal Use Case}

\author{Hassan Jameel Asghar, Paul Tyler and Mohamed Ali Kaafar\\\\
\small Data61, CSIRO, Australia\\
\small \{\texttt{hassan.asghar, paul.tyler, dali.kaafar}\}\texttt{@data61.csiro.au}
}

\date{\today}

\maketitle
\begin{abstract}
This document describes the application of a differentially private algorithm to release public transport usage data from Transport for New South Wales (TfNSW), Australia. The data consists of two separate weeks of ``tap-on/tap-off'' data of individuals who used any of the four different modes of public transport from TfNSW: buses, light rail, train and ferries. These taps are recorded through the smart ticketing system, known as Opal, available in the state of New South Wales, Australia.
\end{abstract}
\section{Introduction}
The Opal smart ticketing system\footnote{See \url{https://www.opal.com.au}.} has been introduced by the (state) Government of New South Wales (NSW) in Australia as a means to use NSW's public transport services. The Opal card can be used on different transport modes managed by the Transport for New South Wales (TfNSW), including trains, buses, light rail and ferries. A tap on using one of the Opal card readers at a location (train stations, buses, ferry wharfs and light rail stops) marks the beginning of a trip and a subsequent tap-off records the end of the trip. The data collected from the Opal card contains a rich set of information that could be used to gain valuable insights about transport usage across NSW. The data if released could bring potential benefits to commuters, businesses and others. However, if released in its raw form the data can have many privacy violations such as the ability to track someone's trip or series of trips. For this reason, TfNSW and Data61/CSIRO engaged in a project to release a privacy-preserving version of a sample of the Opal data as open data downloadable by anyone. We refer to this sample as the Opal dataset in the rest of this document.

The Opal dataset consists of rows each of which indicates a trip from a single individual. A row, among other attributes, contains a unique identifier for the Opal card, the trip date, transport mode (bus, train, etc.), tap-on time, tap-off time, tap-on location and tap-off location. The tap-on and off times have the granularity of one minute. The tap-on and off locations are the precise locations of train stations, bus stops, ferry wharfs and light rail stops. The Opal dataset consists of trips from two separate one week periods; a total of 14 days of data. To release Opal data in a privacy-preserving manner, we used the notion of differential privacy~\cite{calib-noise}. Differential privacy gives a rigorous definition of privacy that is tied to the process purported to provide privacy. Informally, the definition states that a process (an algorithm) is differentially private if the output is indistinguishable if any single individual's data is removed from the dataset. The notion of indistinguishability is probabilistic in nature and is quantified with the help of a parameter denoted $\epsilon$, often called the privacy budget. We prefer the use of differential privacy over ad hoc ways of de-identifying datasets, since the latter is a practice that can only check the (in)effectiveness of known attacks, and could be susceptible to new attacks or future sources of background knowledge available to the attacker\footnote{We define an attacker as anyone who has access to the output data, which also covers accidental re-identification of individuals in the released dataset(s).} about individuals in the dataset. In contrast, differential privacy is immune to new attacks\footnote{This guarantee follows from the definition of differential privacy. Some attacks might still be applicable either due to the improper implementation of differential privacy~\cite{holes} or due to threat models that are not captured within the definitional framework~\cite{diff-priv-fire}. For instance, the attack mentioned in the latter citation exploits the fact that there could be time differences in answering different (online) queries depending on the dataset, thus releasing some side channel information. This is not captured by the definition of differential privacy (the analyst only observes answers to the queries and not other effects of the queries such as the time taken to answer the query). Both attack categories have their counterparts in the field of cryptography as well.} or any future sources of knowledge. The released Opal dataset is available for download at~\cite{opal-data}.

In what follows we first give a brief background on differential privacy in Section~\ref{sec:background}, slightly tailored to the specific case of the Opal dataset. In Section~\ref{sec:sbh}, we describe the differentially private algorithm used to release the Opal dataset, its privacy and utility guarantees and why it was chosen as the algorithm suitable for the Opal dataset. In Section~\ref{sec:app-on-opal}, we describe how the dataset was processed and released, the differentially private nature of the whole process, and the privacy parameters used to release the dataset.  We conclude in Section~\ref{sec:conc}.
\section{Background}
\label{sec:background}
We assume a domain, i.e., data universe, $\domain$. An instance $x \in \domain$ shall be called a \emph{point}. A dataset $D$ is an $n$-tuple, each row of which is a point $x \in \domain$.\footnote{The distinction between a point and a row is important. A dataset may contain multiple rows that are equal to the same point $x \in \domain$.} We denote this as $D \in \domain^n$. In this document, we assume each row of $D$ to represent a trip from an individual. The attributes of the domain $\domain$ include, among other attributes, tap-on times and locations, tap-off times and locations, date of the trip and transport mode. Notice that there is no distinction between sensitive or non-sensitive attributes. Thus, the domain may include potentially identifying attributes such as the unique ID of an individual making the trip. We denote a row of $D$ by small case letters, e.g., $x$ and $y$. For a row $x$ the value of the $i$th attribute is denoted $x_i$. We say two $n$-row datasets $D$ and $D'$ are neighbours, denoted $D \sim D'$ if they differ in only one row. That is, both datasets contain $n$ rows, $n - 1$ of which are the same in the two datasets.

\begin{definition}[Differential privacy~\cite{calib-noise, dp-book}]
A randomized algorithm (mechanism) $\mathcal{M}: \domain^n \rightarrow R$ is $(\epsilon, \delta)$-differentially private if for every $S \subseteq R$, and for all neighbouring databases $D, D' \in \domain^n$, the following holds \[ \mathbb{P} (\mathcal{M}(D) \in S ) \le e^{\epsilon} \mathbb{P} (\mathcal{M}(D') \in S) + \delta. \] If $\delta = 0$, then $\mathcal{M}$ is $\epsilon$-differentially private.
\end{definition}

The parameter $\delta$ is required to be a negligible function of $n$ \cite[\S 2.3, p. 18]{dp-book},\cite[\S 1.6, p. 9]{salil-tut}.\footnote{A function $f$ in $n$ is negligible, if for all $c \in \mathbb{N}$, there exists an $n_0 \in \mathbb{N}$ such that for all $n \ge n_0$, it holds that $f(n) < n^{-c}$.} The parameter $\epsilon$ on the other hand should be small but not arbitrarily small. We may think of $\epsilon \le 0.01$, $\epsilon \le 0.1$ or $\epsilon \le 1$~\cite[\S 1]{perplexed}, \cite[\S 3.5.2, p. 52]{dp-book}. When $\delta = 0$, the resulting notion is sometimes referred to as \emph{pure} differential privacy in contrast to a non-zero $\delta$, which is called \emph{approximate} differential privacy. In this document, we shall use differential privacy as an umbrella term for both pure and approximate differential privacy. Where a distinction is required, we shall specifically indicate as such.

\begin{remark}[Trip privacy]
\label{rem:trip}
Since each row of the dataset $D$ represents a trip, our notion of privacy is related to \emph{trip privacy}. That is, in the above definition of differential privacy two neighbouring datasets differ in one trip. An individual may have multiple trips in the dataset. Thus, in this sense we ensure that the addition or removal of a single trip of an individual does not change the output of the differentially private algorithm beyond what is allowed by the privacy parameter $\epsilon$. This is akin to providing edge privacy as opposed to node privacy in the context of graph data~\cite[\S 1]{graph-privacy}; the latter being a stronger notion of privacy. Later on we discuss why we chose trip privacy as our preferred notion. In more detail, we in fact first partition the Opal dataset according to date and transport modes, so that for each transport mode and date pair, we obtain a separate partition (See Section~\ref{sec:app-on-opal}). This then amounts to providing privacy of a trip of an individual on a particular date and transport mode. 
\end{remark}

The following definition is useful in interpreting $(\epsilon, \delta)$-differential privacy~\cite[\S 1.6]{salil-tut}. 

\begin{definition}
Two random variables $Y$ and $Y'$ taking on values in $R$ are called $(\epsilon, \delta)$-indistinguishable if for every $S \subseteq R$, the following holds 
\[
\mathbb{P} (Y \in S ) \le e^{\epsilon} \mathbb{P} (Y' \in S) + \delta.
\]
\end{definition}

With the above definition, we can interpret $(\epsilon, \delta)$-differential privacy as $\epsilon$-differential privacy with probability at least $1 - \delta$ in light of the following theorem~\cite{conc-dp-bun}, \cite{salil-tut}.
\begin{theorem}
Two random variables $Y$ and $Y'$ are  $(\epsilon, \delta)$-indistinguishable if and only if there are events $\mathsf{E}$ (in the sample space of $Y$) and $\mathsf{E}'$ (in the sample space of $Y'$) such that 
\begin{itemize}
 \item $\mathbb{P} ( \mathsf{E}), \mathbb{P} ( \mathsf{E}') \ge 1 - \delta$, and
 \item $Y$ given $\mathsf{E}$ and $Y'$ given $E'$ are $(\epsilon, 0)$-indistinguishable.
\end{itemize} \qed
\end{theorem}
An important property of differential privacy is that it composes~\cite{dp-book}.
\begin{theorem}[Basic composition]
\label{the:basic-comp}
If $\mathcal{M}_1, \ldots \mathcal{M}_k$ are each $(\epsilon, \delta)$-differentially private then $\mathcal{M} = ( \mathcal{M}_1, \ldots, \mathcal{M}_k)$ is $(k\epsilon, k\delta)$-differentially private. \qed
\end{theorem}
The above is sometimes referred to as (basic) \emph{sequential} composition as opposed to parallel composition defined next.
\begin{theorem}[Parallel composition~\cite{mcsherry-par-comp}]
\label{the:par-comp}
Let $\mathcal{M}_i$ each provide $(\epsilon, \delta)$-differential privacy. Let $\domain_i$ be arbitrary disjoint subsets of the domain $\domain$. The sequence of $\mathcal{M}_i(D \cap \domain_i)$ provides  $(\epsilon, \delta)$-differential privacy, where $D \in \domain^n$.\qed
\end{theorem}
Note that partitioning the domain $\domain$ into disjoint sets means partitioning a dataset $D$ into disjoint sets of rows and \emph{not} disjoint sets of columns. Indeed partitioning into disjoint sets of rows means that addition and removal of a row changes only one of the disjoint sets. 

A query is defined as a function $q: \domain^n \rightarrow \mathbb{R}$. 
\begin{definition}[Counting queries and point functions]
A \emph{counting} query is specified by a predicate $q: \domain \rightarrow \{0, 1\}$ and extended to datasets $D \in \domain^n$ by summing up the predicate on all $n$ rows of the dataset as
\[
q(D) = \sum_{x \in D} q(x).
\]
A point function~\cite{salil-tut} is the sum of the predicate $q_y: \domain \rightarrow \{0, 1\}$, which evaluates to $1$ if the row is equal to the point $y \in \domain$ and $0$ otherwise, over the dataset $D$. Note that computing all point functions, i.e., answering the query $q_y(D)$ for all $y \in \domain$, amounts to computing the histogram of the dataset $D$.
\end{definition}

\begin{definition}[Global sensitivity~\cite{dp-book, salil-tut}]
The global sensitivity of a counting query $q: \domain^n \rightarrow \mathbb{N}$ is 
\[ 
\Delta q = \max_{\substack{D, D' \in \domain^n \\ D \sim D'}} \lVert q(D) - q(D') \rVert_{1}, 
\] 
where $\lVert \cdot \rVert_{1}$ is the $l_1$-norm.\footnote{For a vector $\mathbf{x}$, the $l_1$-norm $\lVert \mathbf{x} \rVert_1$ is defined as $\sum_i | x_i |$, where $x_i$  is the $i$th element of $\mathbf{x}$ and $i$ ranges over all elements of $\mathbf{x}$.} \qed
\end{definition}
It is easy to see that the global sensitivity of a counting query $q: \domain^n \rightarrow \mathbb{N}$ is $1$.

\begin{definition}[Local sensitivity~\cite{dp-book, salil-tut}]
The local sensitivity of a counting query $q: \domain^n \rightarrow \mathbb{N}$ at the dataset $D \in \domain^n$ is 
\[ 
\Delta q = \max_{\substack{ D \sim D'}} \lVert q(D) - q(D') \rVert_{1}.
\]
\end{definition}
The difference between local and global sensitivity is that whereas the latter is defined over all possible neighbouring datasets $D, D'$, the former is only evaluated at neighbours of a \emph{fixed} dataset $D$, i.e., the input dataset. A key property of differential privacy is that it is immune to post-processing. That means that a released dataset from a differentially private algorithm cannot be further processed to reduce privacy (without access to the original dataset). 

\begin{theorem}[Post-processing~\cite{dp-book}]
\label{the:post-proc}
If $\mathcal{M} : \domain^n \rightarrow R$ is $(\epsilon, \delta)$-differentially private and $f : R \rightarrow R'$ is any randomized function, then $f \circ \mathcal{M} : \domain^n \rightarrow R'$ is $(\epsilon, \delta)$-differentially private.\qed
\end{theorem}

The Laplace mechanism is employed as a building block in the main algorithm we use to generate the (privacy-preserving) Opal Dataset.
\begin{definition}[Laplace mechanism~\cite{calib-noise}]
\label{def:laplace}
The Laplace distribution with mean $0$ and scale $b$ has the probability density function 
\[ 
\text{Lap}(x \mid b) = \frac{1}{2b}e^{-\frac{|x|}{b}}.
\]
We shall remove the argument $x$, and simply denote the above by $\text{Lap}(b)$. Let $q: \domain^n \rightarrow \mathbb{R}$ be a query. The mechanism \[ \mathcal{M}_{\text{Lap}}(q, D, \epsilon) = q(D) + \text{Lap}\left(\frac{\Delta q}{\epsilon}\right) \] is known as the Laplace mechanism.\footnote{Recall that $\Delta q$ denotes the sensitivity of the query $q$ (global or local).} The Laplace mechanism is $\epsilon$-differentially private~\cite[\S 3.3]{dp-book}. Furthermore, with probability at least $1 - \beta$~\cite[\S 3.3]{dp-book} \[ \max_{q \in Q} | q(D) - \mathcal{M}_{\text{Lap}}(q, D, \epsilon) | \le \frac{\Delta q}{\epsilon} \ln \left(\frac{1}{\beta} \right) \doteq \alpha, \] where $\beta \in (0, 1]$. If $q$ is a counting query, then $\Delta q = 1$, and the error $\alpha$ is \[ \alpha = \frac{1}{\epsilon} \ln \left(\frac{1}{\beta} \right) \] with probability at least $1 - \beta$. \qed
\end{definition}

From the CDF of the Laplace distribution, we see that for any $t \ge 0$, 
\begin{equation}
\label{eq:lap-bound}
\mathbb{P} \left( \text{Lap}(b) \ge tb \right) = \frac{1}{2} \exp (-t) \le \exp (-t)
\end{equation}

The following definition relates to the output distribution of point functions on a dataset $D$.
\begin{definition}
\label{def:dense}
We say that a dataset $D$ is $(k, \gamma)$-dense if a $\gamma$ fraction of points $x \in D$ have $q_x(D) \ge k$. A dataset is said to be dense, if the values of $k$ and $\gamma$ are large, and sparse otherwise.
\end{definition}
Note that we have deliberately not made the second part of the definition precise. This is because we shall only use the above definition to informally interpret the output of the algorithm (described next), for which the relative difference between dense and sparse datasets suffices.

\section{The Stability-based Histogram Algorithm}
\label{sec:sbh}
In this section we describe the algorithm we used to generate differentially private Opal data. The algorithm, which appears in~\cite[\S 7.1]{salil-tut},\cite{balcer}, slightly modifies the online version of the algorithm proposed in~\cite{bun-stable}, and generates a differentially private synthetic dataset (approximately) answering all point functions $q_y$, where $y \in \domain$. That is, the algorithm can be used to approximate the histogram of a dataset $D$. We represent the algorithm in Algorithm~\ref{alg:algo-diff-hist}, and then discuss the intuition behind the algorithm followed by its privacy and utility guarantees. We call this algorithm the stability-based histogram algorithm, following~\cite[\S 3.3]{salil-tut}. Shortening further, we shall refer to this algorithm as the SBH algorithm from now onwards.

\begin{algorithm}[h]
\caption{\texttt{Stability-based Histogram}~\cite{bun-stable, salil-tut, balcer}}
\label{alg:algo-diff-hist}
\SetAlgoLined
\SetCommentSty{mycommfont}
\SetAlCapSkip{1em}
\DontPrintSemicolon{}
\SetKwInOut{Input}{Input}
\let\oldnl\nl
\newcommand{\nonl}{\renewcommand{\nl}{\let\nl\oldnl}}
\Input{Domain $\domain$, dataset $D \in \domain^n$, parameters $\epsilon$ and $\delta$.}
Initialize $D_{\mathsf{out}} \leftarrow \emptyset$.\;
\For{each point $x \in \domain$}{
	\textbf{if} $q_x(D) = 0$ \textbf{then} set $a_x = 0$. \;
	\If{$q_x(D) > 0$}{Set $a_x \leftarrow q_x(D) + \text{Lap} (\frac{2}{\epsilon})$.\;
		\If{$a_x < 2  \ln\left( \frac{2}{\delta} \right)/\epsilon + 1$ }{
		Set $a_x \leftarrow 0$.\;
		}
		\Else{
		Set $a_x \leftarrow \text{round}(a_x)$.\;
		}
	}
	\textbf{if} $a_x > 0$ \textbf{then} append $a_x$ copies of $x$ to $D_{\mathsf{out}}$.
}
Output $D_{\mathsf{out}}$.
\end{algorithm}

\subsection{Some Remarks about the SBH Algorithm}
The algorithm leverages on the answers to the point functions being stable in the neighbourhood of the dataset $D$. That is if the local sensitivity of the query does not change after removing a few rows from $D$, we can consider the answers to those queries as being stable. The ``few'' in the previous sentence is characterized precisely via the threshold set in step 6 of the SBH algorithm. We shall refer to the right hand term in the inequality as the threshold. If the dataset is dense\footnote{Recall our definition of dense and sparse datasets (Definition~\ref{def:dense}).} then the answers to most point functions on points in the dataset are expected to be above the threshold, and hence will be present in the output (after perturbation via the Laplace mechanism). However, this might not hold in general for all datasets. Consider for instance a dataset $D$ each row of which contains a unique ID. This means that the answer to each point function on the dataset is either $1$ (if the ID is present in the dataset), or $0$ (if the ID is not present in the dataset). With overwhelming probability (depending on $\delta$) the output will be an empty dataset.\footnote{From a utility point of view this is still good, since the answer to each point function is either $0$ or $1$, and an empty dataset means that every point function had a count less than the threshold.} Thus, in order to get some non-trivial output from the algorithm, the dataset should not be too sparse. Hence, in principle, the algorithm can be used for any dataset, even those that contain identifying information. It is simply the case that such identifying information will result in the algorithm not reproducing those rows in the output (by virtue of being sparse). In order to ensure that the algorithm produces some non-trivial synthetic dataset, we therefore removed some of the information that would make the original dataset sparse and aggregated other information to improve density. We will give more details on this in Section~\ref{sec:app-on-opal}.

\subsection{Privacy and Utility of the SBH Algorithm}

The privacy of the SBH algorithm is established in the following theorem whose proof we reproduce from~\cite{bun-stable, salil-tut} for completeness.
\begin{theorem}
\label{the:sbh-priv}
The SBH algorithm is $(\epsilon, \delta)$-differentially private~\cite{bun-stable, salil-tut}.
\end{theorem}
\begin{proof}
Notice that it suffices to argue privacy for the answers $a_x$ to the queries $q_x$, since the output dataset $D_{\mathsf{out}}$ can be constructed through the queries $q_x$ and their answers $a_x$ via post-processing. Consider two neighbouring datasets $D$ and $D'$, where $D'$ is obtained from $D$ by replacing row $x$ with row $x'$. Then the only point queries whose answers differ in the two datasets are $q_x$ and $q_{x'}$. Since each point query $q_y$ is independent of any other point query, we can look at the two aforementioned queries separately. After proving privacy for each, we can use composition to argue privacy for both together. 

Consider the answers $a_x(D)$ and $a_x(D')$ to the query $q_x$ on the datasets $D$ and $D'$, respectively. Since the row $x$ is in $D$, we know that $q_x(D) > 0$. If it is also true that $q_x(D') > 0$ then $a_x(D)$ and $a_x(D')$ are $(\epsilon/2, 0)$-indistinguishable due to the Laplace mechanism, where the subsequent steps 6 and 7 maintain $(\epsilon/2, 0)$-indistinguishability by the closure under post-processing property of differential privacy. If on the other hand we have $q_x(D') = 0$, then $a_x(D')$ is always $0$. But $q_x(D) = 1$, since $D$ and $D'$ are neighbouring datasets that differ in only one row, i.e., $x$, and therefore agree on all other rows. Now, we want to show that the probability that $a_x(D) \neq 0$ is bounded by $\delta/2$. We call the event $a_x(D) \neq 0$ as the \emph{bad event}. This is only possible if
\begin{align*}
q_x(D) + \text{Lap}\left(\frac{2}{\epsilon}\right) &\ge \frac{2}{\epsilon}  \ln\left( \frac{2}{\delta} \right) + 1 \\
\Rightarrow  1 +  \text{Lap}\left(\frac{2}{\epsilon}\right) &\ge \frac{2}{\epsilon}  \ln\left( \frac{2}{\delta} \right) + 1 \\
\Rightarrow   \text{Lap}\left(\frac{2}{\epsilon}\right) &\ge \frac{2}{\epsilon}  \ln\left( \frac{2}{\delta} \right).
\end{align*}
Therefore, by setting $b = \frac{2}{\epsilon}$ and $t =  \ln\left( \frac{2}{\delta} \right)$ in Eq.~\ref{eq:lap-bound}, we get
 \begin{align*}
 \mathbb{P}( a_x(D) \neq 0 ) &= \mathbb{P} \left( \text{Lap}\left(\frac{2}{\epsilon}\right) \ge \frac{2}{\epsilon}  \ln\left( \frac{2}{\delta} \right) \right) \\
&\le \exp\left( - \ln\left( \frac{2}{\delta} \right) \right) \\
& = \frac{\delta}{2}
\end{align*}
Thus the answers $a_x(D)$ and $a_x(D')$ are $(0, \delta/2)$-indistinguishable. Together with the previous case, by composition, answers $a_x(D)$ and $a_x(D')$ are $(\epsilon/2, \delta/2)$-indistinguishable. By symmetry, the same holds for the answers $a_{x'}(D)$ and $a_{x'}(D')$. For all other point functions the answers are identically distributed in the two datasets. Therefore, by composition all answers are $(\epsilon, \delta)$-indistinguishable. The rounding step in Step 9 is simple post-processing. Hence, the mechanism is $(\epsilon, \delta)$-differentially private.
\end{proof}
The following theorem states the utility of the SBH algorithm.
\begin{theorem}
All point queries on the synthetic dataset $D_{\mathsf{out}}$ have error at most $O \left( \frac{\ln \frac{2}{\delta}}{\epsilon} \right)$.
\end{theorem}
\begin{proof}
See~\cite[\S 3.3, p. 24]{salil-tut}.
\end{proof}

\subsection{Why the SBH Algorithm?}
When applied to the Opal dataset, the SBH algorithm had several advantages that we highlight in the following. 

\subsubsection{Computational Efficiency}
The SBH algorithm has the advantage that the run time is polynomial in $n$, i.e., the number of rows in the dataset, as opposed to polynomial in the size of the domain $\domain$, which in turn is exponentially large in the number of attributes. In fact, the domain size $| \domain |$ can become substantially large with only a few attributes of high cardinality. For instance, consider a domain $\domain$ having 4 attributes, each with cardinality $1000$. Then the size of the domain, i.e., $| \domain |$ is already $10^{12}$. To see why the SBH algorithm has run time polynomial in $n$, note that we need not run through each point $x \in \domain$ mentioned in step 2 of the algorithm. Instead, we need only run through points that exist in the dataset (since the point functions of all the points not in the dataset $D$ have by definition answers of $0$). 

However, care must be taken with this modification of the algorithm for reasons of privacy. More specifically, we need to ensure that the rows of the dataset $D$ are either randomly shuffled or sorted in some prescribed order independent of the data, e.g., lexicographic order. This can be done either before or after the application of the (modified) algorithm. To see why order is important, consider two neighbouring datasets $D$ and $D'$ that only differ in the first row. In dataset $D$, the first row is $x$, and in dataset $D'$ the first row is $x'$. The rest of the rows in the two datasets are identical (both in order and value). Now assume that $q_x(D)$ and $q_{x'}(D')$ are both non-zero and significantly greater than the threshold $\frac{2}{\epsilon}  \ln\left( \frac{2}{\delta} \right)$ of the SBH algorithm. Then, with overwhelming probability the first few rows in the output dataset from $D$ will be equal to the point $x$ and the first few rows in the output dataset from $D'$ will be equal to the point $x'$. Thus, the two outputs are not $(\epsilon, \delta)$-indistinguishable for any reasonable values of $\epsilon$ and $\delta$. It is easy to see how this is avoided when the original dataset $D$ is sorted according to some fixed order independent of the dataset. Or, equivalently, via a random shuffle. 

\subsubsection{Introducing New Points}
Another feature of the SBH algorithm is that it does not introduce new points in its output. That is, it only outputs those points that exist in the original dataset $D$, and any points not in $D$, but in the domain, are not present in the output. This is, of course, done with the drawback that some points of $D$ are stripped from the output. In essence, points with low or zero counts are both mapped to zero counts. It was one of the desired characteristics of the synthetic Opal dataset that no new points be added. Without this requirement, a straightforward mechanism to output the histogram of some dataset $D$, i.e., answering all point functions, is to apply Laplace noise of magnitude $1/\epsilon$ to the answer $a_x(D)$ of each point $x \in \domain$ (even points that are not present in $D$). Of course all possible points in the (unrestricted) domain might include points that have no meaning in the real world. For instance, a point might correspond to a trip whose route is not in service at that particular time. In a post-processing step, these anomalies can potentially be ruled out. But doing so in general is not easy, as each row needs to be checked for consistency. In the absence of rules to automatically check for consistency, this becomes a prohibitive task. Nevertheless, the algorithm could be used if a dictionary of all possible trips (tap locations and time combinations) is provided.\footnote{One may also simply leave these ``anomalies'' in tact as a natural by product of the synthetic data generation algorithm. But this is not desirable from a usability perspective. An analyst using the dataset does not know whether a given trip is valid or not without further knowledge of all possible valid trips.} However, such a modification now relies on the domain size, and depending on the size, this algorithm can become prohibitive. 


\section{Application to the Opal Dataset}
\label{sec:app-on-opal}
In this section we describe how we applied the SBH algorithm on the Opal dataset and our choice of privacy parameters.

\subsection{Producing the Output}
\label{sub:output}
Looking at the SBH algorithm, it is clear that as long as the given dataset is sufficiently dense, the output will contain many points from the original dataset. In order to make the dataset more dense we performed a number of pre-processing steps. Note that these steps were taken to increase the density of the dataset(s) to control the impact on utility only. While steps similar to these are used in other settings in an attempt to provide privacy, these do not need to be interpreted as improving privacy in our case, as privacy is automatically guaranteed by the process being differentially private. As discussed below, we actually partitioned the dataset into disjoint partitions. After that we applied the SBH algorithm on several combination of columns (attributes) from each partition to release several differentially private synthetic datasets. The following is the summary of the steps taken. 
\begin{enumerate}
	\item We removed any unique identifiers. These by definition make the dataset sparse.
	\item We binned the tap-on and off times to within a 15 minute window.
	\item The number of possible tap-on and tap-off locations for buses (i.e., bus stops) lead to sparseness. We therefore collapsed them into postcodes.
	\item We partitioned the dataset into $4 \times 14 = 56$ disjoint datasets. This was done by first splitting the dataset by the four transport modes and then separating the datasets belonging to different days.
	\item For each partitioned dataset, we only retained tap-on and tap-off times and locations (the date and mode of transport is now a constant).
	\item Fix a partitioned dataset $D_i$, where $i$ indicates the partition number. We further partitioned the dataset into 6 different datasets according to the combination of columns. These 6 combinations were: (1) tap-on times, (2) tap-on locations, (3) tap-off times, (4) tap-off locations, (5) tap-on times and locations, and (6) tap-off times and locations. We label these datasets $D_{i, j}$, where $1 \le j \le 6$. Note that these new partitions are not disjoint.
	\item We ran the SBH algorithm on each dataset $D_{i, j}$ to obtain the corresponding synthetic dataset, where $1 \le i \le 56$ and $1 \le j \le 6$. The parameters used for the SBH algorithm for the dataset $D_{\cdot, j}$ were $\epsilon_j$ and $\delta_j$, subject to the condition that $\sum_{j = 1}^{6} \epsilon_j = \epsilon$ and $\sum_{j = 1}^{6} \delta_j = \delta$. Note that implicitly we are saying that the same parameter pair $(\epsilon_j, \delta_j)$ is used for datasets $D_{i, j}$ for $1 \le i \le 56$. 
\end{enumerate}

We now show that the whole process is differentially private.
\begin{theorem}
The released Opal datasets are $(\epsilon, \delta)$-differentially private.
\end{theorem}
\begin{proof}
First note that steps 1, 2, 3 and 5 are essentially fixing the domain $\domain$ and therefore apply equally to all datasets $D$ from $\domain^n$. Therefore, these steps have no bearing on differential privacy. For the remaining steps, first fix an $i$. Then we see that all outputs from the datasets $D_{i, j}$, $1 \le j \le 6$, are altogether $(\sum_{j = 1}^{6} \epsilon_{j}, \sum_{j = 1}^{6} \delta_{j}) = (\epsilon, \delta)$-differentially private due to the sequential composition theorem (see Theorem~\ref{the:basic-comp}). Now, all outputs $D_{i, j}$, for $1 \le i \le 56$ and $1 \le j \le 6$ are $(\epsilon, \delta)$-differentially private according to the parallel composition theorem (See Theorem~\ref{the:par-comp}), since each set $D_{i, j}$, $1 \le j \le 6$, is $(\epsilon, \delta)$-differentially private. This completes the proof.
\end{proof}

\begin{remark}
Recall our discussion on trip privacy in Remark~\ref{rem:trip}. We note that since any identifying information has already been removed in the output datasets, and since the datasets have been decoupled into different datasets based on transport modes and dates, there is no obvious link between the trips. Therefore, we believe that the somewhat weaker notion of trip privacy, as opposed to all trips from an individual, is reasonable in this case. 
\end{remark}

\subsection{Choice of Parameters}
\label{sub:params}
Our discussion in this section applies to a generic disjoint partition $i$, where $1 \le i \le 56$. We therefore omit the subscript $i$ and will simply refer to the $6$ further partitions of the dataset $D_i$, as $D_1, \ldots, D_6$. Recall that according to our labelling, datasets $D_1, \ldots, D_4$ correspond to the one-way marginal counts, and datasets $D_5$ and $D_6$ correspond to the two-way marginal counts. We set the values of $\epsilon_j$'s as follows
\begin{align*}
\epsilon_1 = \epsilon_2 = \epsilon_3 =  \epsilon_4 &= 1,\\
 \epsilon_5 = \epsilon_6 &= 2.
\end{align*}
The decision of assigning $\epsilon_j$'s is one of balancing privacy against utility. The higher privacy budget assigned to the two-way marginals ensured that more points were present in the output datasets.
This means that the overall value of $\epsilon$ was $8$ (by basic composition). This value of $\epsilon$ is admittedly higher than what is generally considered in the differential privacy literature.\footnote{Although, not by far! For instance, the experimental evaluation of the DualQuery algorithm from Gaboardi et al. uses an $\epsilon$ up to $5$~\cite[\S 6]{dual-query}.} However, due to the removal of sensitive or identifying information, we believe that this is a reasonably trade-off between privacy and utility. 

For $\delta$, recall that in theory $\delta$ is supposed to be smaller than any inverse polynomial power of $n$, where $n$ is the size of the dataset $D$. However, in practice, we need to instantiate $\delta$ with some value. Once we instantiate $\delta$, by definition $\delta$ is a constant, and is no longer equivalent to its theoretical definition. Therefore, in practice we need to ensure that $\delta$ is a small constant (relative to the size of the dataset). By looking at the SBH algorithm, we see that we need to set $\delta$ such that the probability of the \emph{bad event}\footnote{See the second paragraph in the proof of Theorem~\ref{the:sbh-priv} for the definition of the bad event.} is small. This probability is given by $\delta_j/2$. We fixed $\delta_j = 1/8000000 \approx 2^{-23}$, making this probability approximately $2^{-24}$. The value of $\delta_j \approx 2^{-23}$ implies that the overall value of $\delta$ (by basic composition) is less than $2^{-20}$ or $10^{-6}$. This value of $\delta$ has been used before in the literature~\cite[\S 6]{dual-query}, and is slightly far from another value of $2^{-30}$ recommended in~\cite[\S 8, p. 11]{sec-sample}.

\section{Conclusion}
\label{sec:conc}
Privacy-preserving high dimensional data release is a difficult problem. Research in differential privacy has resulted in many proposed algorithms that seek to generate synthetic versions of the data such that they give approximately correct answers to a specific set of queries. However, in many cases such algorithms take time proportional to the data universe, i.e., the size of the data domain, which in turn is exponential in the number of attributes present in the dataset. Thus, in practice many of these algorithms are not viable for data release.\footnote{See~\cite{salil-tut} for a detailed comparison of differentially private algorithms from a computational point of view.} Furthermore, in the case of open data the usability factor is paramount. A synthetic dataset may contain rows that have no meaning in the real world (e.g., a trip via a route that does not exist.). The purpose of open data is to release the data for any one to use, experts and novices alike. For such users, presence of such anomalous rows (which are in fact artefacts of synthetic data) is not desirable. The algorithm to release the Opal dataset was chosen with these two main considerations in mind. It is possible that in the future better algorithms can be identified or developed that allow richer insights into transport data. We believe this release will guide efforts in that direction.

\bibliographystyle{unsrt}
\bibliography{opal-tech-report}

\end{document}